%% file: main.tex
\documentclass[runningheads]{llncs}

\usepackage{geometry}
\usepackage[T1]{fontenc}
\usepackage{graphicx}
\input{definitions.sty}

\input{macros.tex}

\begin{document}
\title{Swim till You Sink: \\Computing the Limit of a Game}
\author{Rashida Hakim\inst{1} \and
Jason Milionis\inst{1} \and Christos Papadimitriou\inst{1}\and
Georgios Piliouras\inst{2}}
\authorrunning{R. Hakim et al.}
\institute{Columbia University \\ (\email{rashida.hakim@columbia.edu}, \email{jm@cs.columbia.edu}, \email{christos@columbia.edu}) \and
Google DeepMind (\email{gpil@deepmind.com})}
\maketitle              %
\begin{abstract}
During 2023, two interesting results were proven about the limit behavior of game dynamics: First, it was shown that there is a game for which no dynamics converges to the Nash equilibria. Second, it was shown that the sink equilibria of a game adequately capture the limit behavior of natural game dynamics. These two results have created a need and opportunity to articulate a principled computational theory of the meaning of the game that is based on game dynamics. Given any game in normal form, and any prior distribution of play, we study the problem of computing the asymptotic behavior of a class of natural dynamics called the noisy replicator dynamics as a limit distribution over the sink equilibria of the game. When the prior distribution has pure strategy support, we prove this distribution can be computed efficiently, in near-linear time to the size of the best-response graph. When the distribution can be sampled --- for example, if it is the uniform distribution over all mixed strategy profiles --- we show through experiments that the limit distribution of reasonably large games can be estimated quite accurately through sampling and simulation.

\keywords{Replicator Dynamics \and Sink Equilibria.}
\end{abstract}
\vspace{-0.8cm}
\begin{figure}[ht]
\centering
\includegraphics[height=10cm,keepaspectratio]{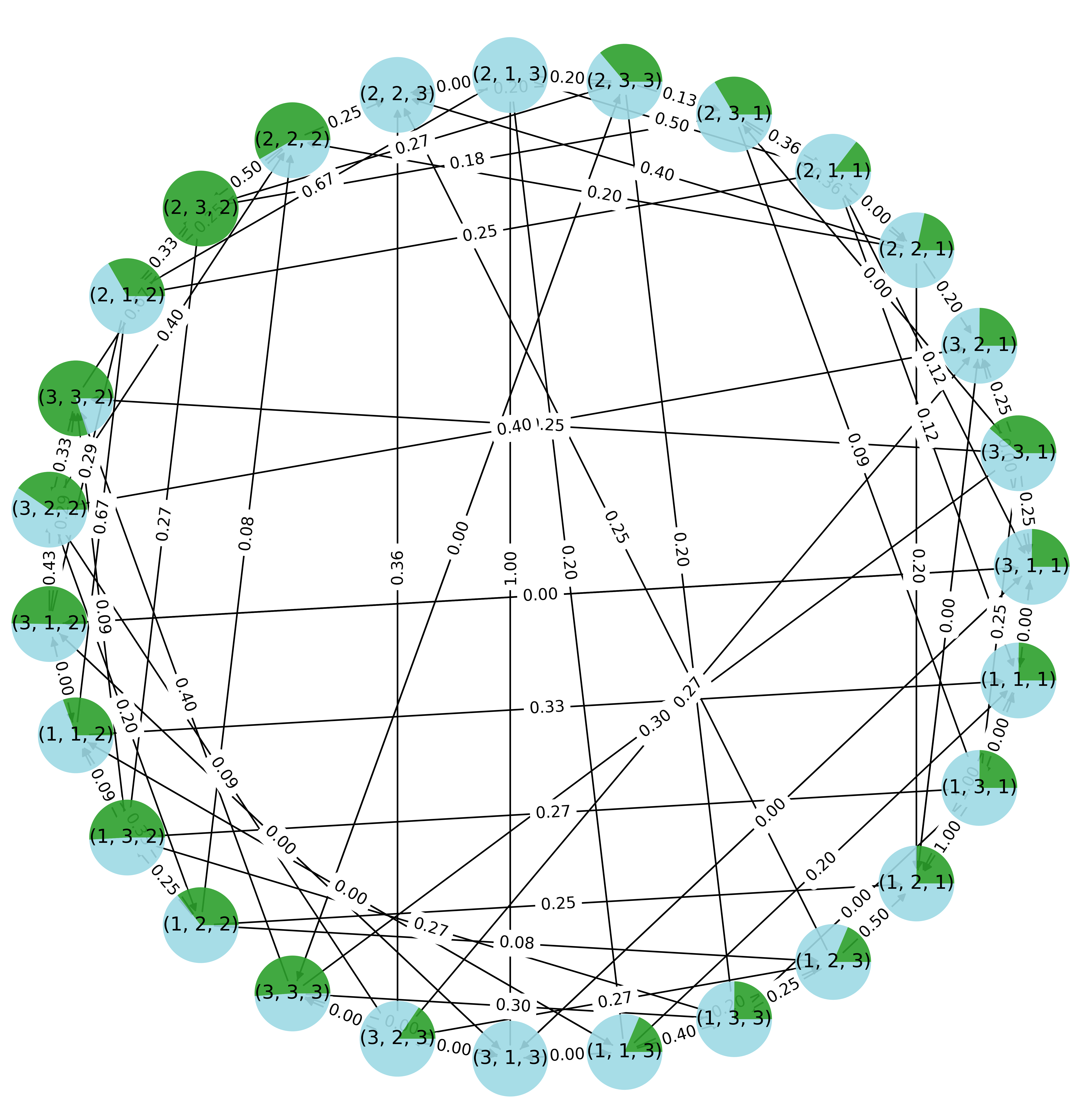}
\vspace{-0.3cm}
\caption{The better-response graph of the $3 \times 3 \times 3$ game depicting the hitting probabilities of the pure profiles as pie charts.}
\label{fig:game:high_pie}
\end{figure}
\input{intro}

\input{prelims}

\input{algo}

\input{experiments}

\input{openquestions}

\bibliographystyle{splncs04}
\bibliography{refs}

\end{document}

%% file: macros.tex
\newcommand{\len}{{{\text{len}}}}

\renewcommand{\Pr}{\mathop{\bf Pr\/}}

\newcommand{\eps}{\epsilon}

\newcommand{\calN}{\mathcal{N}}

\def\<{\langle}
\def\>{\rangle}

%% file: intro.tex
\section{Introduction}

The Nash equilibrium has been the quintessence of Game Theory. The field started its modern existence in 1950 with Nash's definition and existence theorem, and the Nash equilibrium remained for three quarters of a century its paramount solution concept --- the others exist as refinements, generalizations, or contradistinctions. During the past three decades, during which Game Theory came under intense computational scrutiny, the Nash equilibrium has lost some of its appeal, as its fundamental incompatibility with computation became apparent. The Nash equilibrium has been shown intractable to compute or approximate in normal-form games \cite{DGP,CDT,EY}, while its other well-known deficiency of computational nature --- the ambiguity of non-uniqueness and the quagmire of equilibrium selection \cite{HS} --- had already been known.

A long list of {\em game dynamics} --- that is, dynamical systems, continuous- or discrete-time, defined on mixed strategy profiles --- proposed by economists over the decades are all known to {\em fail} to converge consistently to the Nash equilibrium. This included Nash's own discrete-time dynamics used in his proof, the well-known replicator dynamics treated in this paper, and many others. Given this, the following question acquired some importance:

\begin{quote}
{\em Question 1: We know that every game has a Nash equilibrium. But does every game have a dynamics that converges to the Nash equilibria of the game?} 
\end{quote}

A negative answer would be another serious setback for the Nash equilibrium, and impetus would be added to efforts (see for example \cite{papadimitriou2019game,omidshafiei2019alpha}) to elevate the limit behavior of natural game dynamics as a proposed ``meaning of the game,'' an alternative to the Nash equilibrium. An important obstacle to these efforts was that the nature of the limit behavior of natural dynamics in general games had been lacking the required clarity. It had been known for 40 years since Conley's seminal work \cite{conley1978isolated} that the right concept of limit behavior in a general dynamical system is a system of topological objects known as its {\em chain recurrent components.}  However, this concept is mathematically intractable for general dynamical systems: there can be infinitely many such sets, of unbounded complexity.

\begin{quote}
{\em Question 2: Is there a concrete characterization, in terms of familiar game-theoretic concepts, of the chain recurrent sets in the special case of natural game dynamics in normal-form games?} 
\end{quote}

During this past year, there was important progress on both questions.

\begin{enumerate}
    \item It was proven in \cite{MPPS} that the answer of Question 1 above is negative:  there is a game for which no game dynamics can converge to the Nash equilibria --- that is, there is no dynamics such that the fate of all initial strategy profiles are the Nash equilibria, and the Nash equilibria are themselves fixed points of the dynamical system. Thus, the Nash equilibrium is fundamentally incapable of capturing asymptotic player behavior.

    \item Biggar and Shames establish a useful characterization of the chain recurrent sets of natural recurrent dynamics \cite{BS_sccs}: it was shown that each chain recurrent component of a game under the replicator dynamics 
    contains 
    the union of one or more {\em sink equilibria} of the game. Sink equilibria, first defined by Goemans et al. \cite{goemans_sink_equilibria} in the context of the price of anarchy, are the sink strongly connected components of the better response graph of the game.
\end{enumerate}

We believe that these two results open an important opportunity to articulate a new approach to understanding a normal-form game.  Instead of considering it, as game theorists have been doing so far, as the specification of an intractable equilibrium selection problem, we propose to see it as a specification of the limit behavior of the players. According to this point of view \cite{papadimitriou2019game}, a game is a mapping from a prior distribution over mixed strategy profiles (MSPs) to the resulting limit distribution if the players engage in an intuitive and well accepted natural behavior called {\em noisy replicator dynamics,} which is related to multiplicative weight updates and will be defined soon.  This is the quest we are pursuing and advancing in this paper.

\subsection*{Our contributions}
\begin{itemize}
    \item We propose a concrete, unambiguous, and computationally tractable conception of a game as a mapping from any prior distribution over MSPs to the sink equilibria of the game, namely the limit distribution of the noisy replicator dynamics when initialized at the prior.  
    \item We initiate the study of the efficiency of its computation. As a baby step, in the next section we show that the sink equilibria can be computed in time linear in the description of the game.  We also point out that they are intractable for various families of implicit games.
    \item We prove that the mapping from a prior to a distribution over sink equilibria can be calculated explicitly and efficiently (near linear in the size of the game description) when the prior has pure strategy support. This is highly nontrivial because the better response graph of the game may contain many directed cycles of length two with infinitesimal transition probability $\eps$, corresponding to tie edges; the analysis must be carried out at the $\eps\rightarrow 0$ limit. The algorithm involves a number of novel graph-theoretic concepts and techniques relating to Markov chains, and the deployment of near-linear algorithms for directed Laplacian system solving as well as a dynamic algorithm for incrementally maintaining the strongly connected components (SCCs) of a graph.
    \item We also show through extensive experimentation that the general case (arbitrary prior) can be solved efficiently for quite large games.
\end{itemize}

\subsection*{Related work}

{\bf Non-convergence of learning dynamics in games.}
The difficulty of learning dynamics to converge to Nash equilibria in games is punctuated by a plethora of diverse negative results spanning numerous disciplines such as game theory, economics, computer science and control theory~\cite{BaileyEC18,daskalakis2010learning,paperics11,galla2013complex,papadimitriou2019game,vlatakis2020no,andrade2021learning,andrade2023no,hart2003uncoupled,babichenko2012completely,hsieh2020limits,mertikopoulos2017cycles,cheung2021online,young2007possible}.
Recently, 
\cite{MPPS} capped off this stream of negative results with a general impossibility result showing that there is no game dynamics that achieve global convergence to Nash for all games, a result that is independent of any complexity theoretic or uncoupledness assumptions on the dynamics.   Besides such worst case theoretical results, detailed experimental studies suggest that chaos is commonplace in game dynamics, and emerges even in low dimensional systems across a variety of game theoretic applications~\cite{sanders2018prevalence,sato2002chaos,palaiopanos2017multiplicative,piliouras2023multi,chotibut2020route,bielawski2021follow,leonardos2023optimality,pangallo2017}. 

{\bf Dynamical systems for learning in games.}
This extensive list of non-equilibrating results has inspired a program for linking game theory to  dynamical systems \cite{Entropy18,papadimitriou2019game} and Conley's fundamental theorem of dynamical systems \cite{conley1978isolated}.
These tools have since been applied in multi-agent ML settings such as developing novel rankings as well as training methodologies for agents playing games  \cite{omidshafiei2019alpha,rowland2019multiagent,muller2019generalized,omidshafiei2020navigating}.  
Finally, Peyton Young's paper on conventions \cite{peytonyoung} is an important precursor of our point of view in the Economics literature, focusing in the special case of games in which the sink SCCs are pure strategy equilibria.

{\bf Sink equilibria.} The notion of sink equilibrium, a strongly connected component with no outgoing arcs in the strategy profile graph associated with a game, was introduced in \cite{goemans_sink_equilibria}. They also defined an analogue to the notion of Price of Anarchy~\cite{poa}, the Price of Sinking, the ratio between the worst case value of a sink
equilibrium and the value of the socially optimal solution.  The value of a sink equilibrium is defined as the expected
social value of the steady state distribution induced by a random walk on that sink. Later work established further connections between Price of Sinking and Price of Anarchy via the $(\lambda,\mu)$-robustness framework~\cite{rough09}. A number of negative, PSPACE-hard complexity results for analyzing and approximating sink equilibria have been established in different families of succinct games~\cite{fabrikant2008complexity,mirrokni2009complexity}.

Finally, \cite{hassin1992mean} compute the limiting stationary distribution of an irreducible MC with vanishing edges; their technique can be used in our framework to solve for the time averaged long run behavior within a sink SCC.

%% file: prelims.tex
\section{Preliminaries}
We assume the standard definitions of a normal-form game $G$ with $p$ players with pure strategy sets $\{S_i\}$, and their utilities $U_i.$  We denote by $|G|$ the size of the description of $G$.  The {\em better-or-equal response graph} $B(G)$ has the pure strategy profiles as nodes, and an edge from $u$ to $v$ if $u$ and $v$ differ in the strategy of only one player $i$, and $U_i(v)\geq U_i(u)$. Let $E$ be the set of edges of $B(G)$. Notice that, because of the tie edges and the transitive edges, the number of edges in the response graph can be much larger than the size of the description of the game, i.e., $|E|=\Omega(|G|)$. The {\em sink equilibria} of $G$ are the sink strongly connected components (sink SCCs) of $B(G)$, that is, maximal sets of nodes with paths between all pairs, such that there is no edge leaving this set.  We shall define many other novel graph-theoretic concepts for this graph in the next section.  Our first theorem delineates the complexity of finding the sink equilibria of a game:

\begin{theorem}
The sink equilibria can be computed in time near-linear in the description of the game presented in normal form, whereas computing them in a graphical game is PSPACE-complete. 
\end{theorem}

\begin{proof}
The first claim follows from the fact that, even though (as pointed out in the preceding paragraph) $B(G)$ has more edges than the size of the description of the game, there is an equivalent graph of linear size with the same transitive closure (and therefore strongly connected components) obtained as follows: For each player and each pure strategy profile for other players, consider the subgraph of only the nodes that correspond to the strategy profile for the other players together with some action for the player. Sort the nodes by their outcome for the player. For each node in order from lowest to highest outcome, create an edge to the next node in increasing order as well as the last node of the same outcome if the next node has a higher outcome (only if this would not be a self-loop). This preserves transitive closure as compared to $B(G)$ and each node has at most outgoing 2 edges (it sums to $\leq 3/2$ edges per node on average in the worst case). Besides sorting this can be done in linear time.

The second claim follows from known results \cite{Fabrikant08_Complexity}; the result holds for other forms of succinct descriptions of games, such as Bayesian or extensive form games.
\end{proof}

Next we define the {\em noisy replicator dynamics} on $G$  \cite{papadimitriou2019game}, a noisy generalization of the classical replicator dynamics \cite{Schuster1983533}.  It is a function mapping the set of MSPs of $G$ to itself as follows:

\begin{itemize}
    \item $\phi(x) = \partial G(x + \eta\cdot\rm{BR}_x + \calN_x(0, \delta)) $, where 
    \item BR$_x$ is the unit best response vector at $x$ {\em projected to the subspace of $x$} --- that is, containing zeros at all coordinates in which $x$ is zero; this ensures that the support of $x$ never increases;
    \item $\calN(0,\delta)_x$ is Gaussian noise, { also projected};
    \item the function $\partial G$ maps $(x+ \eta\cdot\rm{BR}_x + \calN(0, \delta))_x$ either to itself, if it is inside the domain of the game's MSPs, or to the closest point in the support's boundary, otherwise; 
    \item and $\delta, \eta >0$ are important small parameters.
\end{itemize}

\paragraph{Justification.} The replicator dynamics \cite{taylor1978evolution,Schuster1983533} has been for four decades the standard model for the evolution of strategic behavior.  In connection with Economics and Game Theory, it has the important advantage of invariance under positive affine changes in the players' utilities.  For our purposes, it is approximated via the noisy version of Multiplicative Weights Update (MWU) \cite{Arora05themultiplicative}. Projecting the noise to the support of the current MSP $x$ is motivated by evolution and extinction, and is instrumental for fast convergence. This precise dynamics has been used extensively in reinforcement learning for game play, see for example \cite{omidshafiei2019alpha,omidshafiei2020navigating}.

Finally, we define a dynamics on the {\em pure} strategy profiles (that is, a Markov chain), called the \emph{Conley-Markov Chain} of $G$, or CMC$(G)$ \cite{papadimitriou2019game,goemans_sink_equilibria}.  If $(u,v)$ is an edge of $B(G)$ corresponding to a defection of player $i$, its probability in CMC$(G)$ is proportional to $U_i(v)-U_i(u)$, with the edges out of each node $u$ normalized to one.  It is not hard to see that this is the limit of the noisy replicator dynamics as the noise goes to zero and the MSP goes to $u$.  Importantly, however, CMC$(G)$ also has an infinitesimal probability $\eps$ for each tie edge. (Note that this probability is used symbolically as it descends to zero, and, in the interest of clarity, it does not affect the normalization at $u$.) This treatment of tie edges reflects two things: First, it was shown in \cite{biggar2023attractor} that tie edges must be included in the calculation of the sink equilibria for their theorem to hold; and second, to incorporate tie edges in a way compatible with Conley's theorem \cite{conley1978isolated} is to think of them as conduits of a {\em balanced random walk} on the undirected edge between the two nodes in which the MSP is changing via tiny steps of $\sigma$ at a time, so that it will take $\Theta({1/ \sigma^2})$ steps for the transition to be completed, justifying its infinitesimal transition probability.

%% file: algo.tex
\section{A Combinatorial Algorithm for the Hitting Probabilities}
\label{sec:algo}

We start by collapsing all sink SCCs of CMC$(G)$ to single absorbing nodes. Our main goal is to compute the hitting probabilities %
from each node $i$ of $CMC(G)$ to each of the absorbing nodes --- that is, the probability that a path starting from $i$ will end up in the node --- albeit {\em in the limit as $\eps\rightarrow 0$.}  %
The hitting probabilities can be defined in two equivalent ways, both of which we will use in our proofs. Define $h_{iS}$ to be the hitting probability of node $i$ to sink $S$, that is, the probability that a path from $i$ will eventually be absorbed by $S$. Let $p_{ij}$ be the transition probability of node $i$ to node $j$ (the weight of the edge $(i, j)$ or 0 if there is no edge). Then the hitting probabilities are the smallest non-negative numbers that satisfy the following system of equations. 
\begin{align}\label{hp1}
  h_{iS}=
  \begin{cases}
    \sum\limits_{j\in\text{nodes}} p_{ij}h_{jS}, & \text{if $i\not\in S$}.\\
    \hfil 1,                    %
    & \text{if $i\in S$}
  \end{cases}
\end{align}

If we define $\Psi_{iS}$ as the (potentially infinite) set of paths that start at $i$ and end at some node in $S$, then we equivalently have the following set of equations.
\begin{align}\label{hp2}
h_{iS} = \sum_{p \in \Psi_{iS}}\Pr[p]
\end{align}

In this section we prove the following:
\begin{theorem}
The limit hitting probabilities of CMC$(G)$ can be computed in time $O(|E|^{\frac{4}{3}})$, where $E$ is the set of edges of CMC$(G)$.
\end{theorem}

Significant progress has been made in solving linear systems associated with weighted directed graphs such as Equation \ref{hp1} faster than the time required to solve arbitrary linear systems. %
Two problems can now be solved in almost-linear time in the size of the graph: computing the stationary distribution of an irreducible Markov chain (hence abbreviated MC); and computing the escape probabilities \cite{cohen2017almost}. The computation of escape probabilities in a random walk maps directly to the problem of computing hitting probabilities in a MC. So we have fast algorithms for our problem in the case of no tie edges. However, the introduction of tie edges creates an ill-conditioned problem, and we are interested in its solution as $\epsilon\rightarrow 0$. A possible approach would be to solve the system of equations given in Eq. \ref{hp1} symbolically and then take limits as $\epsilon\rightarrow0$; however, solving large systems of equations symbolically is intractable. Instead, we take a combinatorial approach to transform any given CMC$(G)$ into a simpler MC that preserves the limit hitting probabilities of the original CMC$(G)$ but eventually has no tie-edges. The hitting probabilities of this simplified MC can be computed in almost linear time as mentioned above.%

\subsection{Outline of the Algorithm}
\begin{enumerate}
    \item The input to the algorithm is CMC$(G)$ --- actually, it could be any $\epsilon$-MC $M$ with absorbing nodes. The output is the list of hitting probabilities  $\{h_{iS}\}$, the probabilities that the sink SCCs of the graph will be reached by each of the nodes in the rest of the graph.
    \item We start by collapsing the sink SCCs of $M$.  
    \item We calculate the SCC's of $M$ {\em without} the $\epsilon$-edges, called {\em rSCC's}. This makes sense since $\epsilon$ edges are traversed at a far slower rate than the rest.
    \item Next we must handle a phenomenon called a {\em pseudosink}, an rSCC that only has $\epsilon$-edges outgoing. Within a pseudosink, the MC achieves convergence to a steady state before exiting, and therefore all of its nodes have the same hitting probabilities. The pseudosinks are identified one by one and collapsed, with their outgoing $\epsilon$-edges replaced by regular edges in accordance with Def. \ref{weight}.  A simple disjoint set data structure can track the original vertices through the collapses.
    \item A complication is that the collapsed pseudosinks acquire new regular edges to the rest of the graph, and as a result the rSCCs of the graph  must be recalculated. This procedure may also create new pseudosinks, so steps 3 and 4 are repeated until no more pseudosinks exist. 
    \item Once all pseudosinks have been removed this way, any remaining $\epsilon$-edges do not affect the hitting probabilities and can therefore be deleted.  At this point, the hitting probabilities can be computed in almost linear time.
\end{enumerate}

\subsection{Definitions}
\begin{definition}
\textbf{$\epsilon$-Markov Chain: } A $\epsilon$-Markov Chain ($\eps$-MC) is a Markov chain that has two types of edges: regular edges which have weights $c_r - c_{re}\epsilon$ for constants $c_r > 0$ and $c_{re} \geq 0$ and $\epsilon$ edges which have weights $c_e\epsilon$ for constant $c_e > 0$. Thus, the CMC$(G)$ is an $\epsilon$-MC.  As the values of the coefficients $c_{re}$ in the regular edges do not affect the limiting hitting probabilities we shall ignore them. 
\end{definition}
\begin{definition}
\textbf{Sink SCC:} A sink SCC $S$ is a maximal set of nodes that are strongly connected (including connected via $\epsilon$-edges) that has no outgoing edges.
\end{definition}
\begin{definition}
\textbf{rSCC:} A rSCC is a maximal set of nodes that is strongly connected via regular edges. An rSCC may contain $\epsilon$-edges between nodes within the rSCC, but every node is reachable from every other node without requiring $\epsilon$-edges. 
\end{definition}
\begin{definition}
\textbf{Pseudosink:} A pseudosink $P$ is a rSCC that has at least one outgoing $\epsilon$-edge and no outgoing regular edges.
\end{definition}
\begin{definition}
\textbf{Order:} The order of a node $i$, $\textsc{Order}(i)$, is the minimum number of tie edges that exist on a path from $i$ to \textit{any} sink SCC. The maximum order of the current MC, \textsc{MaxOrder}$(M)$, is our gauge of progress in the algorithm. 
\end{definition}
\begin{definition}\label{weight}
The {\bf weight} of a new regular edge from a collapsed pseudosink to node $y$ is as follows (here $O$ is the set of outgoing $\epsilon$-edges from the pseudosink $P$). 
$$W(P, y) = \frac{\sum_{e\in O: e = (x, y)}c_{e}\pi_{P}[x]}{\sum_{e' = (x', y')\in O}c_{e'}\pi_P[x']},$$
where $\pi_{P}[x]$ is the steady-state probability of $x$ within $P$, computed using only regular edges.
\end{definition}

\subsection{Algorithm Correctness}
Throughout the algorithm, we maintain a MC that we denote $M$, initially the MCM$(G)$ with all sink SCCs collapsed.  There are two aspects to validate. The first is that the algorithm progresses until $M$ has no remaining $\epsilon$-edges. The second is that $M$ maintains the property that at all stages the limit hitting probabilities of the original nodes are maintained through collapsing pseudosinks (step 2) and deleting $\epsilon$-edges. %

\subsubsection{Algorithm Progression} 
\begin{lemma} If \textsc{MaxOrder}$(M) \geq 1$ then $M$ contains a pseudosink.
\end{lemma}
\begin{proof} Consider the set of all nodes $i$ that achieve $\textsc{Order}(i) = \textsc{MaxOrder}(M)$. By definition, from this set of nodes, the MC cannot reach any other nodes without using $\epsilon$-edges. Consider the rSCC decomposition of this set of nodes. The regular edges between rSCCs induce a directed acyclic graph on this set of nodes. Every finite DAG has at least one leaf, defined as a vertex that has no outgoing edges. Each leaf is a rSCC with no outgoing regular edges and therefore is a pseudosink.
\end{proof}
\begin{lemma}\label{order reduction} Collapsing all pseudosinks in $M$ reduces the maximum order of $M$ by at least 1.
\end{lemma}
\begin{proof} 
Again, consider the set of all nodes $i$ that that have $\textsc{Order}(i)$ equal to $\textsc{MaxOrder}(M)$ and the DAG representing the rSCC structure of this set. Each leaf of the DAG is a pseudosink. From each pseudosink, there exists a path to some sink SCC that achieves the order $\textsc{MaxOrder}(M)$. Collapsing the pseudosink replaces all outgoing $\epsilon$-edges with regular edges and therefore this same path must now achieve the order $\textsc{MaxOrder}(M) - 1$. Since every rSCC in the DAG is either a leaf or has a regular path to a leaf, all nodes that previously achieved $i$ will now have $\textsc{Order}(i) \leq \textsc{MaxOrder}(M) - 1$. So the maximum order of $M$ is reduced by at least 1.
\end{proof}
Combining these two lemmas means that at each stage of our algorithm we will find one or more pseudosinks and collapse them, decreasing the maximum order by at least 1. So the algorithm will progress until the maximum order reaches 0, at which point we delete all remaining $\epsilon$-edges and are left with a Markov Chain with only regular edges.

\subsubsection{Pseudosink Collapse}

\begin{lemma}\label{tie edge coefficient} Let $P$ be a pseudosink. Let $O$ be the set of outgoing $\epsilon$-edges from $P$. Let $L_e$ be the event that a Markov chain started at any $i\in N_P$ will take $(e = (x, y)) \in O$, where $W_M(e) = c_e\epsilon$ is the weight in $M$ of edge $e$. Then
$$\lim_{\epsilon\rightarrow 0}\Pr[L_e] = \frac{c_{e}\pi_{P}[x] }{\sum_{e'\in O}c_{e'}\pi_{P}}$$ 
This immediately implies that $W(P, y)$ from Definition \ref{weight} is the probability that $y$ is the first node outside of $P$ that a chain started at any $i\in N_P$ will travel to.
\end{lemma}
\begin{proof}
Define $L_e$ as the event that $e \in O$ is the first outgoing $\epsilon$-edge. Also, define $T_L$ as the random variable representing at what timestep the chain leaves $P$. We will start by using the fact that we must leave $P$ via some first outgoing edge. %
\begin{align*}
&1 = \sum_{e\in O}\Pr[L_e]
=\sum_{e\in O}\sum_{t=0}^{t_1} \Pr[L_e \cap T_L = t] + \sum_{e\in O}\sum_{t=t_1 + 1}^{\infty} \Pr[L_e \cap T_L = t]
\end{align*}
Let $d$ be the periodicity of the pseudosink ignoring $\epsilon$-edges. Set $t_1 = \frac{1}{\sqrt{\epsilon}} + k$ where $k$ is an integer random variable chosen uniformly at random from 0 to $d - 1$. Define a constant $L_{max} = \max_{i\in N_M} \sum_{e = (i, j) \in E} c_e$ where $E$ is the set of epsilon edges in $M$. So we have that $L_{max}\epsilon$ upper bounds the probability of taking an $\epsilon$ edge at any particular timestep.
\begin{align*}
&\lim_{\epsilon\rightarrow 0}\sum_{e'\in O}\sum_{t=0}^{t_1} \Pr[L_e' \cap T_L = t]
\leq \lim_{\epsilon\rightarrow 0} \sum_{t=0}^{t_1} L_{max}\cdot\epsilon = \lim_{\epsilon\rightarrow 0} (\frac{L_{max}\epsilon}{\sqrt{\epsilon}} + L_{max}\epsilon k)= 0
\end{align*}

Let $\pi_{P}$ be the stationary distribution on the pseudosink $P$ using only regular edges. We can ignore the internal $\epsilon$-edges because the pseudosink is an rSCC and therefore the regular edges uniquely define the limiting stationary distribution \cite{hassin1992mean}. 

Define $X_t$ as the state of the MC at time $t$. For fixed $t' \geq 0$ and $x \in N_P$, at time $t = t_1 + t'$ we have $\Pr[X_t = x | T_L \geq t] = \pi_{P}[x] + n(t_1)$ where the randomness is over both the MC process and the choice of $k$. This is because while the MC has not left $P$, it is approaching the steady state of $P$ (or its subsequences are if $P$ is periodic). We have that $\vert n(t_1)\vert \leq a_Gb_G^{t_1}$ for some finite constants $a_P$ and $0 \leq b_G<1$ independent of $\epsilon$. This is due to the rate of convergence of subsequences of periodic Markov Chains \cite{bowerman1977convergence}. \\

Noting that the left hand side of our previous equation is independent of $\epsilon$, we can take the limit as $\epsilon\rightarrow0$.
\begin{align*}
1 &= \lim_{\epsilon\rightarrow 0}\sum_{e\in O}\sum_{t=t_1 + 1}^{\infty} \Pr[L_e \cap T_L = t]
= \lim_{\epsilon\rightarrow 0}\sum_{e\in O}\sum_{t=t_1 + 1}^{\infty}c_e\epsilon\Pr[X_t = x \cap T_L \geq t] \\
&= \lim_{\epsilon\rightarrow 0}\sum_{e\in O}\sum_{t=t_1 + 1}^{\infty}c_e\epsilon\Pr[X_t = x | T_L \geq t]\Pr[T_L \geq t]\\
&=\lim_{\epsilon\rightarrow 0}\sum_{e' = (x', y')\in O}c_{e'}(\pi_{P}[x'] + n(t_1))\sum_{t=t_1 + 1}^{\infty} \epsilon\Pr[T_L \geq t]\\
&\frac{1}{\lim_{\epsilon\rightarrow 0} \sum\limits_{e' = (x', y')\in O}c_{e'}(\pi_{P}[x'] + n(t_1))} = \lim_{\epsilon\rightarrow 0}\sum_{t=t_1 + 1}^{\infty} \epsilon\Pr[T_L \geq t]\\
&\frac{1}{\sum\limits_{e' = (x', y')\in O}c_{e'}\pi_{P}[x']} = \lim_{\epsilon\rightarrow 0}\sum_{t=t_1 + 1}^{\infty} \epsilon\Pr[T_L \geq t]
\end{align*}
Now we will compute $\lim_{\epsilon\rightarrow 0}\Pr[L_e]$.
\begin{align*}
&\Pr[L_e] =\sum_{t=0}^{t_1} \Pr[L_e \cap T_l = t] + \sum_{t=t_1 + 1}^{\infty} \Pr[L_e \cap T_l = t]
\end{align*}
We can show that the first term vanishes in the limit as $\epsilon\rightarrow0$.
\begin{align*}
\sum_{t=0}^{t_1} \Pr[L_e \cap T_l = t] &< \sum_{e'\in O}\sum_{t=0}^{t_1} \Pr[L_e' \cap T_L = t]\\
\lim_{\epsilon\rightarrow 0}\sum_{t=0}^{t_1} \Pr[L_e \cap T_l = t] &\leq \lim_{\epsilon\rightarrow 0}\sum_{e'\in O}\sum_{t=0}^{t_1} \Pr[L_e' \cap T_L = t] = 0
\end{align*}
Now we can substitute back in to our expression for $\Pr[L_e]$ and do a similar procedure as above.
\begin{align*}
&\lim_{\epsilon\rightarrow 0}\Pr[L_e] = \lim_{\epsilon\rightarrow 0}\sum_{t=\frac{1}{\sqrt{\epsilon}} + 1}^{\infty} c_{e}\epsilon(\pi_{G}[x] + n(t_1))\Pr[T_L \geq t]\\
&= \lim_{\epsilon\rightarrow 0}c_{e}(\pi_{P}[x] + n(t_1))\cdot\lim_{\epsilon\rightarrow 0}\sum_{t=\frac{1}{\sqrt{\epsilon}} + 1}^{\infty} \epsilon\Pr[T_L \geq t]\\
&=\frac{c_{e}\pi_{P}[x]}{\sum_{e' = (x', y')\in O}c_{e'}\pi_P[x']}
\end{align*}
\end{proof}
\begin{lemma}
Let $P$ be a pseudosink, and let the stationary distribution on $P$ without $\epsilon$-edges be denoted by $\pi_P$. The hitting probabilities to sink SCCs of the overall graph are not affected by collapsing $P$ to a single node $A_P$ with outgoing edges corresponding to $W(A_P, y)$.
\end{lemma}
\begin{proof}
Let $M$ be the $\epsilon$-MC immediately before collapsing $P$, and let $M'$ be the transformed $\epsilon$-MC. We will show that for an arbitrary $i$ and arbitrary sink SCC $S$, $\lim_{\epsilon\rightarrow0}h_{iS}(\epsilon) = \lim_{\epsilon\rightarrow0}h'_{iS}(\epsilon)$ where $h_{iS}(\epsilon)$ is the hitting probability from $i$ to $S$ in $M$ and $h'_{iS}(\epsilon)$ is the analogous quantity in $M'$.

Let $N_P$ be the set of nodes of pseudosink $P$. Define $\Psi_k$ to be the set of paths in $M$ that start at $i\not\in N_P$ and enter and exit $P$ exactly $k$ times. Define $\Psi'_k$ similarly to be the set of paths in $M'$ that visit collapsed node $A_P$ exactly $k$ times. Then we can write the hitting probabilities as in Eq. \ref{hp2}.
\begin{align*}
h_iS(\epsilon) &= \sum_{k = 0}^\infty \sum_{p \in \Psi_{iS}\cap \Psi_k} \Pr[p] \\
h'_iS(\epsilon) &= \sum_{k = 0}^\infty \sum_{p \in \Psi'_{iS}\cap \Psi'_k} \Pr[p]
\end{align*}
Observe that any path $p \in \Psi_k$ can be parameterized by two $k$ length vectors. The first, denoted $\hat{g} \in (N_P)^k$ represents the sequence of entry locations to $P$ of the path. To clarify, $g_j$ is the first node in $P$ that $p$ visits on its $j^{th}$ time entering $P$. The second vector, denoted $\hat{y} \in (N_M \backslash N_P)^k$, represents the sequence of exit locations from $P$ of the path. So $y_j$ is the first node not in $P$ that $p$ visits on its $j^{th}$ exit from $P$. Let $\Psi_{\hat{g}, \hat{y}} \subseteq \Psi_k$ be the set of paths parameterized by $\hat{g}, \hat{y}$.

For set of nodes $B$ and nodes $i \in B, j$, define $T_{B}[i, j] = \sum_{p\in \Psi_{i, j, B}} \Pr[p]$ where $\Psi_{i, j, B}$ is the potentially infinite set of paths from $i$ to $j$ in $M$ that have the property that all nodes on the path are in the set $B$ with the exception of the last node on the path if and only if $j\not\in B$. Define $T'_{B}[i, j]$ as the analogous quantity for paths in $M'$.

Let $N_M$ be the set of nodes of the $\epsilon$-MC $M$. Define $N_H = N_M \backslash N_P$, so $N_H$ is the set of nodes in $M$ that are not in $P$. For $k = 0$, $\Psi_{is} \cap \Psi_k$ only includes paths that travel through nodes in $N_H$. These nodes and the edges between them are unchanged by the collapse procedure implying that $\sum_{p\in \Psi_{iS} \cap \Psi_0} \Pr[p] = \sum_{p\in \Psi'_{iS} \cap \Psi'_0} \Pr[p]$.

For $k\geq 1$, we can write the sum of probabilities of all paths that start at $i$ and end at $S$, that visit $G$ exactly $k$ times as follows.
\begin{alignat*}{2}
&\sum_{p\in \Psi_{iS}\cap \Psi_k} \Pr[p] = \sum_{\hat{y}}\sum_{\hat{g}}\sum_{p\in \Psi_{iS}\cap \Psi_{\hat{g}, \hat{y}}} \Pr[p]\\
&=\sum_{\hat{y}}\sum_{g_1\in N_P}\sum_{g_2\in N_P}...\sum_{g_k\in N_P}T_{N_H}[i, g_1]\prod_{j=1}^{k-1} \left(T_{N_P}[g_j, y_j]T_{N_H}[y_j, g_{j + 1}]\right)\cdot T_{N_P}[g_k, y_k]T_{N_H}[y_k, S] \\
&=\sum_{\hat{y}}\sum_{g_1\in N_P}T_{N_H}[i, g_1]T_{N_P}[g_1, y_1]\cdot\prod_{j=2}^{k}\left(\sum_{g_j\in N_P}T_{N_H}[y_{j-1}, g_{j}]T_{N_P}[g_{j}, y_{j}]\right)T_{N_H}[y_k, S]
\end{alignat*}
Observe that $T_{N_P}[g, y]$ is the probability that the first state after leaving $P$ is $y$ (after entering $P$ at state $g$) and therefore by Lemma \ref{tie edge coefficient}, for all $g\in N_P$, $\lim_{\epsilon\rightarrow0}T_{N_P}[g, y] = W(A_P, y)$

Also, observe that $\sum_{g\in N_P} T_{N_H}[i, g] = T'_{N_H}[i, A_P]$ because every path ending at some node $g\in N_P$ in $M$ ends at $N_G$ in $M'$ and none of the edges used in these paths are affected by the collapse procedure. We also know that paths in $\Psi'_K$ can be parameterized using a single vector $\hat{y} \in (N_{M'} \backslash A_P)^k$, where $y_j$ is the node visited immediately after visiting $A_P$ for the $j^{th}$ time. 
\begin{align*}
\lim_{\epsilon\rightarrow0}\sum_{p\in \Psi_{iS}\cap \Psi_k}\Pr[p] &= \sum_{\hat{y}} T'_{N_H}[i, A_P]\prod_{j=1}^{k-1} (W(g, y_j) T'_{N_H}[y_j, A_P])W(g, y_k)T'_{N_H}[y_k, S]\\
&= \lim_{\epsilon\rightarrow0}\sum_{p\in \Psi'_{iS}\cap \Psi'_k}\Pr[p]
\end{align*}
Since the summations are equal for all $k$, we can see that the hitting probabilities from nodes $i \not\in N_P$ are preserved in the limit. All that is left to prove is that for $g\in N_P$, we have that $\lim_{\epsilon\rightarrow0} h_{gS}(\epsilon) = \lim_{\epsilon\rightarrow0} h'_{A_PS}(\epsilon)$. Let $H_{gS}$ be the event that we hit absorbing state $S$ given that $M$ started at $g$ (so $h_{gS}(\epsilon) = \Pr[H_{gS}]$). Recall that $L_e$ is the event that $e$ is the first outgoing edge from $P$ that the chain takes.
\begin{align*}
\Pr[H_{gS}] &= \sum_{e\in O} \Pr[H_{gS}\vert L_e]\Pr[L_e]\\
&= \sum_{e = (x, y)\in O} \Pr[H_{yS}]\Pr[L_e]\\
\lim_{\epsilon\rightarrow0}h_{gS} &= \sum_{y \not\in P} W(A_P, y)\lim_{\epsilon\rightarrow0}h_{yS}
\end{align*}
Eq.~\ref{hp1} for calculating the hitting probability of $A_P$ in $M'$ is exactly $h'_{A_PS} = \sum_{y \not\in P} W(A_P, y)\lim_{\epsilon\rightarrow0}h'_{yS}$. Since we have that the limiting hitting probabilities are maintained from $y\not\in N_P$, they are also maintained for $g\in N_P$.
\end{proof}

\subsubsection{Deletion of $\epsilon$-Edges at the Final Step}
\begin{lemma}\label{no tie edges taken} If for all nodes $i \in N_M$ have a regular path (only using regular edges) to an absorbing state, then $\forall i$, $\lim_{\epsilon\rightarrow0}P[V_{i}] = 1$ where $V_i$ is the event that the MC started at $i$ is absorbed before taking any $\epsilon$-edges.
\end{lemma}
\begin{proof} Consider the set of regular edges of $M$ which have weights of the form $c_r - c_{re}\epsilon$. Set $\epsilon \leq \frac{\min(c_r)}{2\max{c_{re}}}$, so that every regular edge has weight at least $w_{min} = \min(c_r)/2$. Now set $C_{min} = w_{min}^{|N_M|}$. Note that $C_{min}$ is a constant independent of $\epsilon$. For every node $i$ there exists a regular path $p_{iS}$ of length at most $|N_M|$ to an absorbing state, which must have probability at least $C_{min}$. In addition, recall that $L_{max}$ is a constant defined in Lemma \ref{tie edge coefficient}, such that $L_{max}\epsilon$ upper bounds the probability of taking an $\epsilon$-edge at any timestep.

We will calculate $P[V_{i}]$ by analyzing a new MC $M'$ that depends on $M$. This new MC begins in the "start" state, with the original MC at some node $i'$ which is not an absorbing state (initially this node is $i$). One step of $M'$ is as follows: It evolves the $M$ from $i'$ for $\len(p)$ steps, where $p$ is the path from $i'$ to an absorbing state that has probability $\geq C_{min}$. If $M$ ends up at some absorbing state (one way to do this is to take path $p$) then $M'$ moves to the absorbing "success" state. If $M$ takes any $\epsilon$ transitions during this evolution, then $M'$ moves to the absorbing "failure" state. If neither of these happen, $M$ will be at some non-absorbing state $i'$ and $M'$ will stay in the start state.

Observe that the probability that $M'$ reaches the success state is exactly $P[V_{i}]$, since $M'$ reaches the success state if and only if $M$ reaches some absorbing state before taking any $\epsilon$-edges. Denote the event that $M'$ reaches the success state by $V'_S$. The edge from the start state to the success state has probability at least $C_{min}$. The edge from the start state to the failure state has probability at most $F(\epsilon) = 1 - (1 - L_{min}\epsilon)^{|N_M|}$. This is because there are $len(p) \leq |N_M|$ steps of $M$ taken by a single step of $M'$ and the chance of taking an $\epsilon$ transition at each of those steps is upper bounded by $L_{max}\epsilon$. 

We can use these bounds on the probabilities of the edges of $M'$ to compute that:
\begin{align*}
\Pr[V'_S] &\geq (1\cdot C_{min}) + (0\cdot F) + (\Pr[V'_S]\cdot(1-C_{min}-F))\\
&\geq \frac{C_{min}}{C_{min} + F(\epsilon)}
\end{align*}
Since $\lim_{\epsilon \rightarrow 0} F(\epsilon) = 0$, $\lim_{\epsilon \rightarrow 0} P[V'_S] = 1$ and therefore $\lim_{\epsilon \rightarrow 0}P[V_{i}] = 1$.
\end{proof}

\begin{lemma}\label{can delete ties} If $\textsc{MaxOrder}(M) = 0$, deleting all $\epsilon$-edges (that are not within a sink SCC) does not affect the hitting probabilities.
\end{lemma}
\begin{proof} By definition of order, if $\textsc{MaxOrder}(M) = 0$, every node in $M$ has a regular path to an absorbing state, so we can apply Lemma \ref{no tie edges taken} to get that for all states $i$, $\lim_{\epsilon\rightarrow0}P[V_{i}] = 1$ where $V_{i}$ is the event that $M$ transitions from $i$ to some absorbing state without taking any $\epsilon$-edges. 
\\
Let $M'$ be the graph with all $\epsilon$-edges removed and all edge weights of the form $c_r + c_{re}\epsilon$ set to $c_r$. Let $h'_{iS}$ be the hitting probability from $i$ to absorbing state $S$ in $M'$. Define $\Psi_r$ to be the set of paths in $M$ that only use regular edges, and $\Psi'_r$ to be the analogous set for $M'$.  Then $\Psi_{iS} \cap \Psi_r$ is the set of paths in $M$ that go from $i$ to $S$ only using regular edges. Define $P$ as the set of all paths in $M$, we have that $\Psi_{iS} \cap (\Psi \backslash \Psi_r)$ is the set of paths in $M$ that go from $i$ to $S$ using one or more $\epsilon$-edges.
\begin{align*}
&h_{iS}(\epsilon) = \sum_{p\in (\Psi_{iS} \cap \Psi_r)}\Pr[p] + \sum_{p\in (\Psi_{iS} \cap (\Psi \backslash \Psi_r))}\Pr[p]
\end{align*}
The probability of taking a path from $i$ to $S$ that has one or more $\epsilon$-edges must be less than the probability of taking a $\epsilon$-edge before absorption (since the first event is a subset of the second). 
\begin{align*}
&\sum_{p\in(\Psi_{iS} \cap (\Psi \backslash \Psi_r))}\Pr[p] \leq 1 - \Pr[V_{iJ}] \\
\lim_{\epsilon\rightarrow0}&\sum_{p\in(\Psi_{iS} \cap (\Psi \backslash \Psi_r))}\Pr[p] \leq \lim_{\epsilon\rightarrow0}(1 - \Pr[V_{i}]) = 0
\end{align*}
We can plug this limit into the hitting probability equation and use the property that, since the only structural difference between $M$ and $M'$ is the $\epsilon$-edges, $P_{iS}\cap P_r = P'_{iS}\cap P'_r$. In addition, for all regular edges $e \in E_M$ we have that $\lim_{\epsilon\rightarrow 0} W_M(e) = W_{M'}(e)$ due to our renormalization. Note that we can interchange limits because each weight on an regular edge converges to a positive constant. 
\begin{align*}
\lim_{\epsilon\rightarrow0}h_{iS}(\epsilon) &= \lim_{\epsilon\rightarrow0}\sum_{p\in(\Psi_{iS} \cap \Psi_r)}\Pr[p]
=\sum_{p\in(\Psi_{iS} \cap \Psi_r)}\prod_{e\in p}\lim_{\epsilon\rightarrow0} W_M(e)\\
&=\sum_{p\in(\Psi'_{iS} \cap \Psi'_r)}\prod_{e\in p} W_{M'}(e) = h'_{iS}
\end{align*}
Deleting the $\epsilon$ edges and re-normalizing the regular edges therefore has no effect on the limiting hitting probabilities.
\end{proof}

 \subsection{Running Time of the Algorithm}
Only steps 4, 5, and 6 of the algorithm have superlinear complexity.  For 4 we need to calculate the steady-state probabilities of the nodes of the pseudosink, because they are needed in the calculation of the weights of the edges leaving the collapsed pseudosink.  For 6, we need to compute the hitting probability in an ordinary graph (no $\eps$-edges).  Both of these problems can be solved in time $O(|E|^{1+\delta})$ for all $\delta >0$ \cite{cohen2017almost}. In Step 5, we do incremental maintenance of (r)SCCs. The fastest known algorithms for incremental SCC maintenance take amortized time $O(E^{1+\delta})$ \cite{chen2023almost}.

%% file: experiments.tex
\section{Experiments}
\label{sec:experiments}

We have implemented our algorithm and experimented with random games with various values of the parameters $p$ (players) and $s$ (strategies per player) ranging for both parameters from $2$ to $12$.  In the next subsection, we present certain examples that exhibit interesting behavior viz. our algorithm.
Since our main message is a new way to view a game as a algorithmic map from a prior to a posterior distribution, 
in the second subsection we demonstrate how this works for various reasonably large games.  Given a prior distribution (typically the uniform distribution over all MSPs), we sample from this distribution and then simulate the noisy replicator.  We repeat until our convergence criteria are satisfied, and output the posterior distribution.  This accomplishes our overarching goal, the empirical computation of the meaning of the game.  We repeat this experiment for larger and larger games, taking this simulation to its practical laptop limits.
The code used to generate the entire section is available at \url{https://jasonmili.github.io/files/gd_hittingprobabilities_code.zip}.

\subsection{Some interesting games and their better-response graphs}

We use the following plotting conventions: in each better-response graph, every node of every sink SCC will be colored with a unique color. Other pure profile nodes of the graph will be depicted as a ``pie'' graph with colored areas that indicate the hitting probabilities towards each of the sink SCCs it reaches, as identified by the former colors (of the sink SCCs). Tie edges ($\eps$ edges) appear in the graph as bidirectional ``0.00'' edges; this is only for plotting convenience.

\subsubsection{$3 \times 3$ Game}
We start with a modified version of a game presented by \cite{papadimitriou2019game} that exhibits two sink SCCs: a directed cycle of length four (corresponding to a periodic orbit in the replicator space) and another that is a single pure profile (corresponding to a strict pure NE); see  \Cref{fig:game:georgios_jason_modified}.

\begin{figure}[h!]
\centering
\begin{minipage}[c]{0.5\textwidth}
    \includegraphics[width=\textwidth]{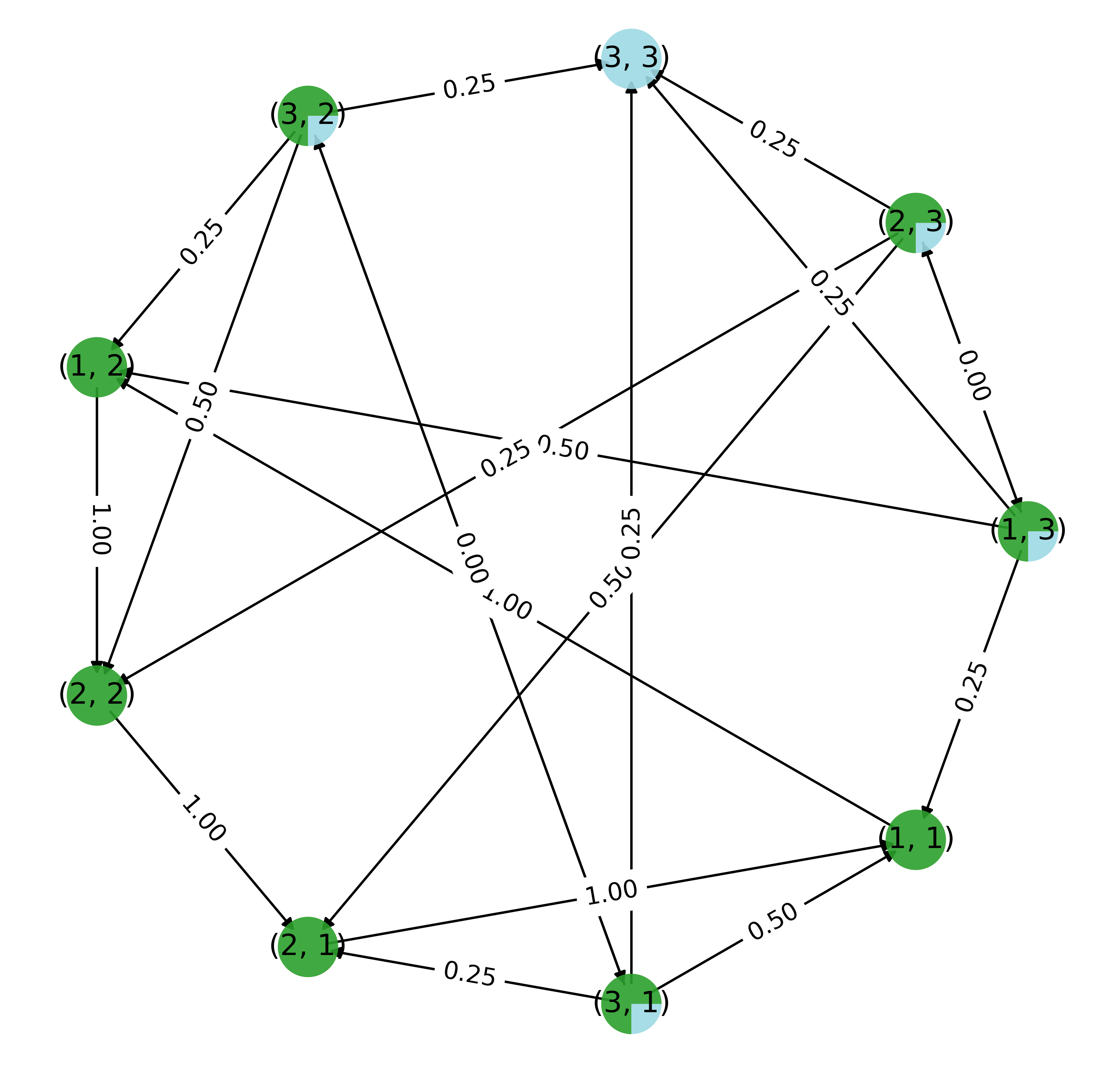}%
\end{minipage}
\hfill
\begin{minipage}[c]{0.29\textwidth}
\[
\begin{pmatrix}
2, 1 & 1, 2 & 0, 0 \\
1, 2 & 2, 1 & 0, 0 \\
0, 0 & 0, 0 & 1, 1 \\
\end{pmatrix}
\]
\end{minipage}
\caption{$3\times 3$ game. Left: the better-response graph. Right: the game utilities.}
\label{fig:game:georgios_jason_modified}
\end{figure}

\subsubsection{Game with Order 1 Profile}
We construct a game, depicted in \Cref{fig:game:rashida_jason} with two sink SCCs, and a pure profile of order 1 (which is also a pseudosink SCC) that needs exactly one tie edge to reach any sink SCC. Notice that the presence of this pure profile affects the hitting probabilities towards the sink SCCs, as described in \Cref{sec:algo}. This example shows that there are cases where a pure NE may not be a sink SCC, or as a matter of fact, not even inside \emph{any} sink equilibrium. That is, this NE is not {\em stochastically stable} in the terminology of \cite{peytonyoung}.
\begin{figure}[ht]
\centering
\begin{minipage}[c]{0.5\textwidth}
    \includegraphics[width=\textwidth]{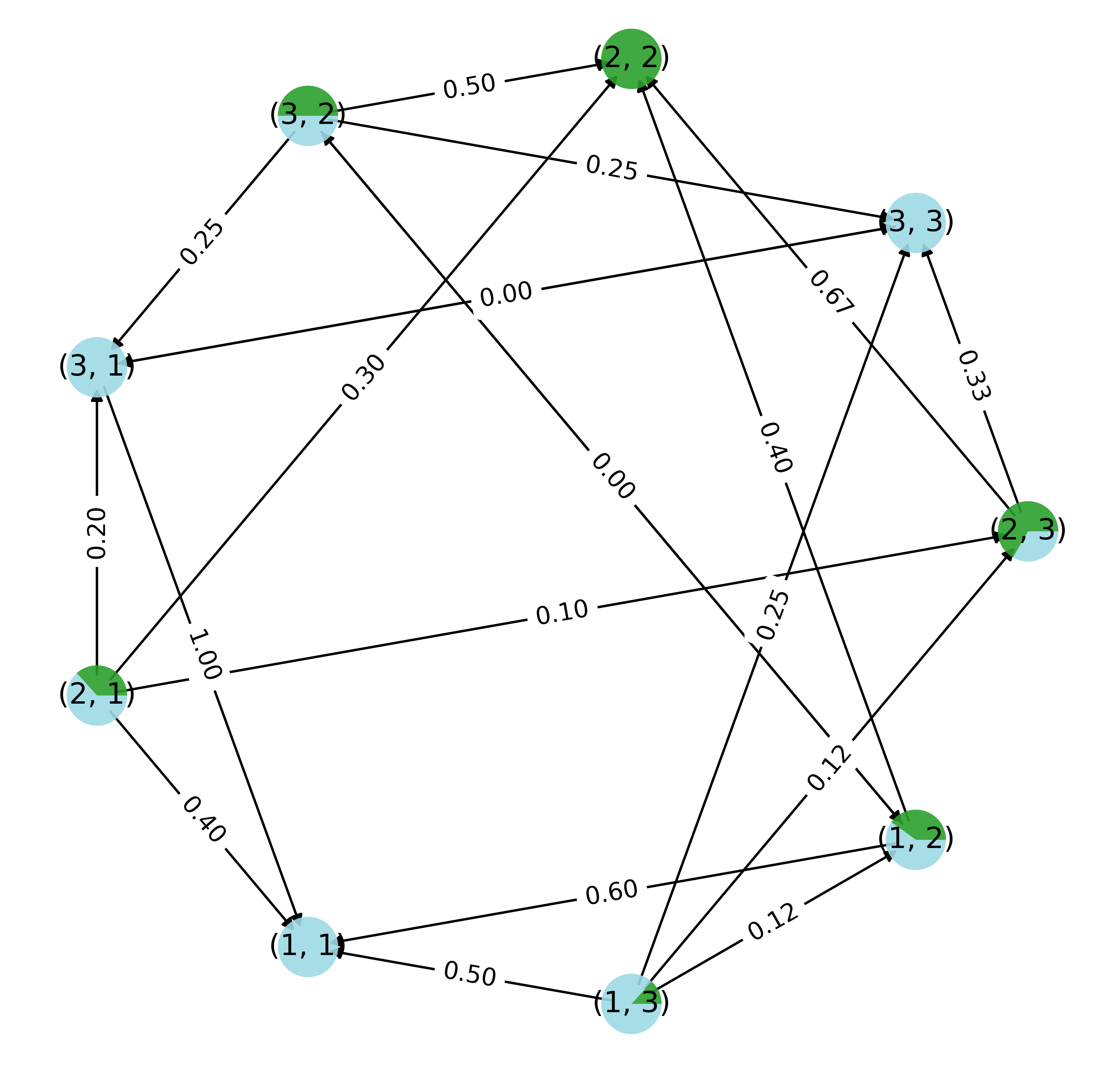}%
\end{minipage}
\hfill
\begin{minipage}[c]{0.29\textwidth}
\[
\begin{pmatrix}
4, 4 & 1, 1 & 0, 0 \\
0, 0 & 3, 3 & 1, 1 \\
2, 2 & 1, 1 & 2, 2 \\
\end{pmatrix}
\]
\end{minipage}
\caption{Tie game. The profile $(3,3)$ is order 1. Left: the better-response graph. Right: the game utilities.}
\label{fig:game:rashida_jason}
\end{figure}

\subsubsection{$3\times 3\times 3$ Game}

The utilities for this game can be found in our code. See \Cref{fig:game:high,fig:game:high_pie}.

\begin{figure}[!ht]
\centering
\subfloat[][The better-response graph of the game; nodes of sink SCCs are depicted in red.]{%
    \includegraphics[width=0.6\textwidth,height=6.5cm,keepaspectratio]{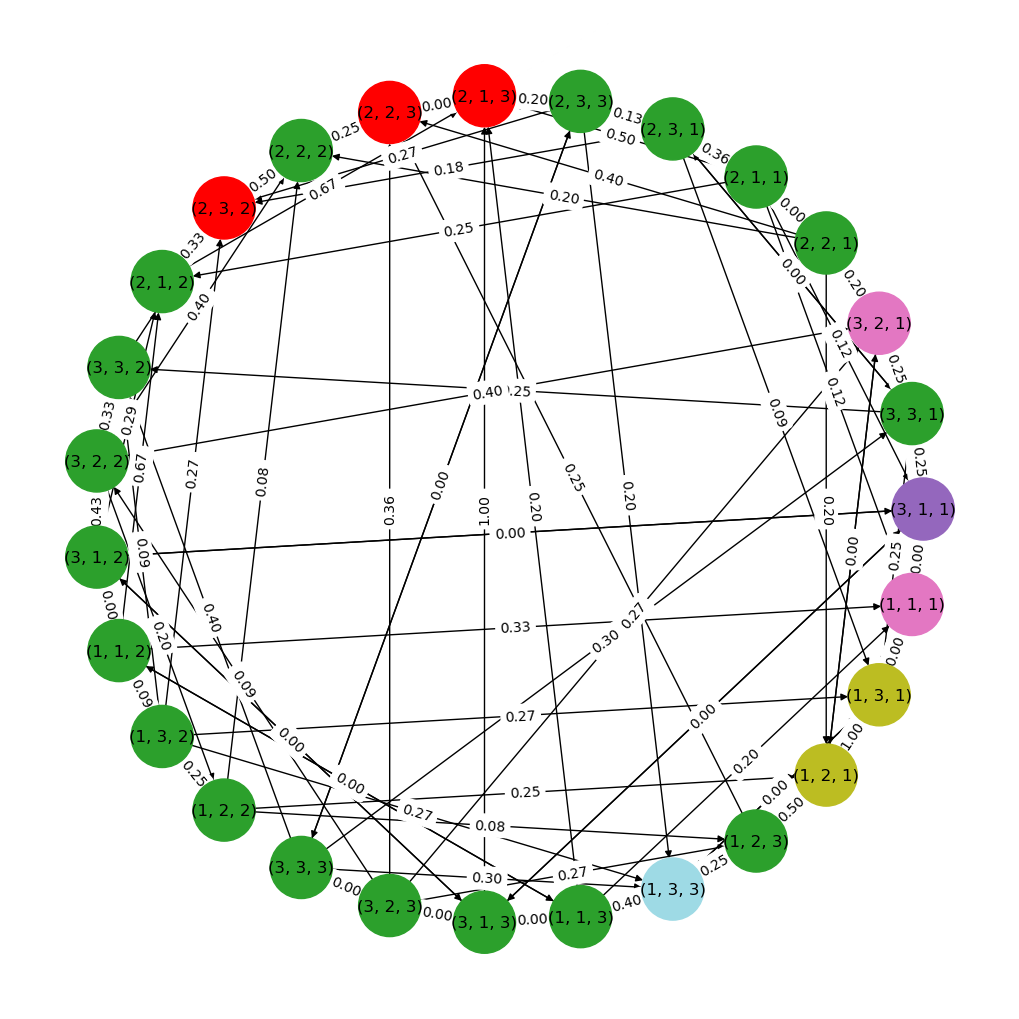}%
}\hfill
\begin{minipage}[b]{.35\textwidth}
    \subfloat[][For added clarity, we show a subgraph of (a) of all pure profiles of order $\ge 1$, along with the sink SCCs.]{%
        \includegraphics[width=0.8\textwidth,height=5cm,keepaspectratio]{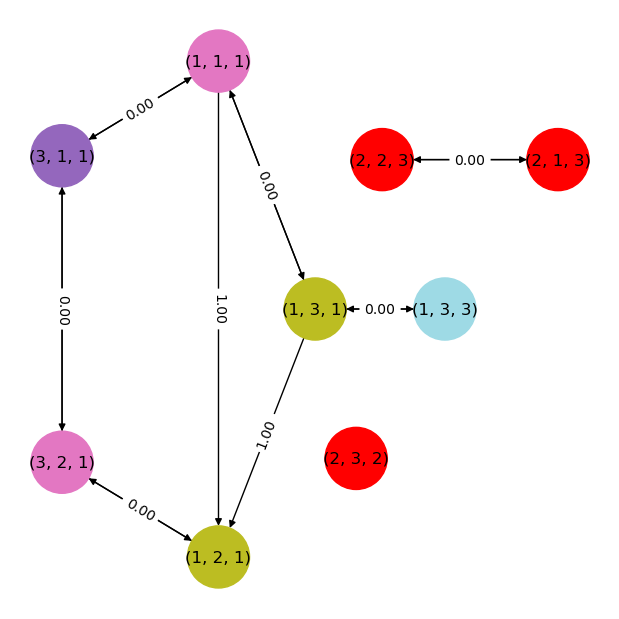}%
    }\vfill
    \subfloat[][The color coding that depicts the orders of the various pure-profile nodes.]{%
        \includegraphics[width=0.8\textwidth,height=1cm,keepaspectratio]{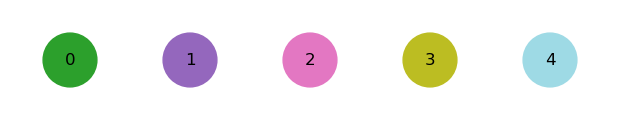}%
    }
\end{minipage}
\caption{The $3 \times 3 \times 3$ game.}
\label{fig:game:high}
\end{figure}

\subsection{Convergence Statistics}
\paragraph{Methodology.} We generate games of various sizes and random utilities (see the figures below), and we carry out a number of independently-randomized experiments of running noisy replicator dynamics (RD) on each. For each sampled point of the prior distribution (typically uniform), we run multiple independently-randomized instances of the noisy RD to obtain an empirical distribution. We consider the outcome of the game as the {\em empirical last-iterate distribution}, i.e., the average of all obtained distributions after run $T$. We keep track of the total variation (TV) distance between the running average distribution (e.g., at time $t<T$) and the ex-post empirical last-iterate (average at time $T$).
We consider that a distribution has achieved good enough convergence when the TV distance is less than 1\% --- we found that this is roughly the accuracy that is feasible in a laptop-like experimental setup. All calculations in this section were performed on an Apple M2 processor with the use of multi-threading with 8 parallel threads.

\begin{figure}[h!]
    \centering
    \includegraphics[width=0.5\textwidth]{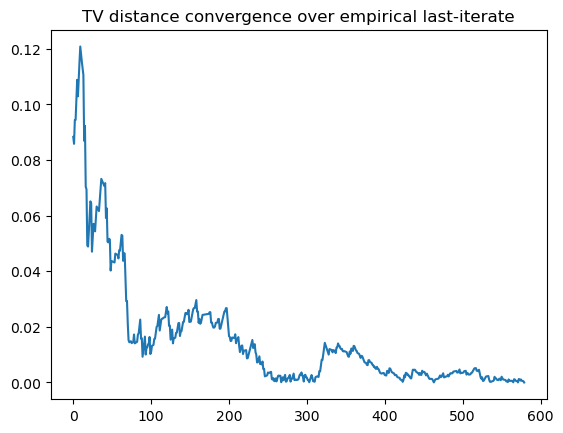}
    \caption{Convergence example. 40 independent runs of noisy RD were used for each sample.}
    \label{fig:high_conv}
\end{figure}

\begin{figure}[h!]
\centering
\subfloat[\label{fig:unif_conv_2s}][2-player, $s$-strategy games.]{%
    \includegraphics[width=0.5\textwidth,height=5.2cm,keepaspectratio]{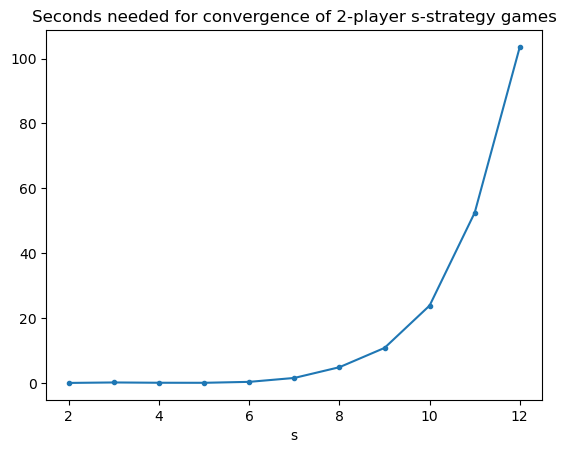}%
}\hfill
\subfloat[\label{fig:unif_conv_3s}][3-player, $s$-strategy games.]{%
    \includegraphics[width=0.5\textwidth,height=5.2cm,keepaspectratio]{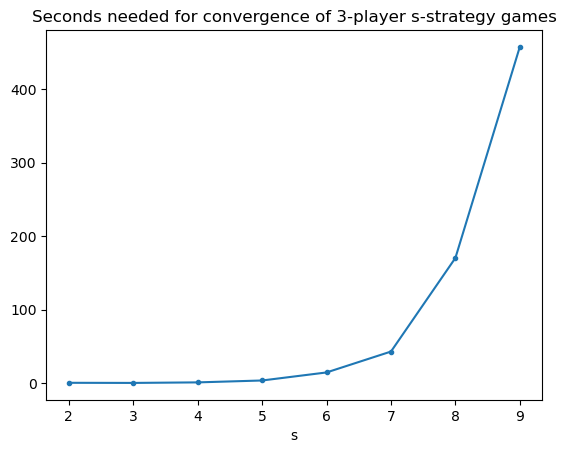}%
}
\end{figure}

\begin{figure}[h!]
\centering
\subfloat[\label{fig:unif_conv_n2}][$n$-player, 2-strategy games.]{%
    \includegraphics[width=0.499\textwidth,height=5.2cm,keepaspectratio]{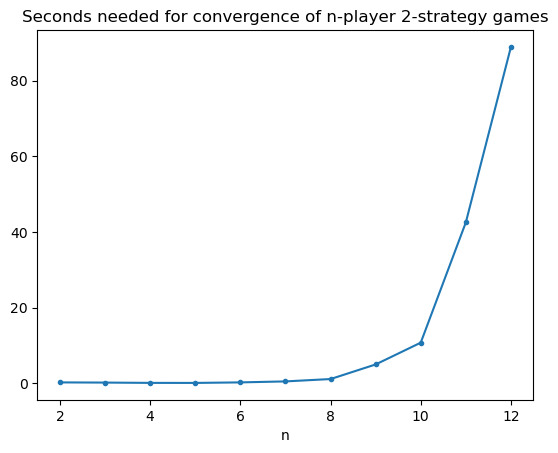}%
}\hfill
\subfloat[\label{fig:unif_conv_n3}][$n$-player, 3-strategy games.]{%
    \includegraphics[width=0.499\textwidth,height=5.2cm,keepaspectratio]{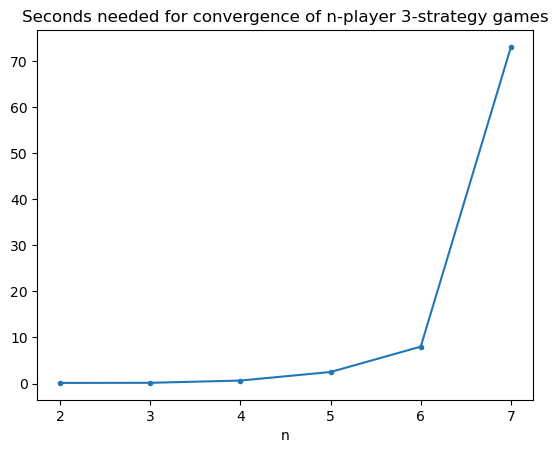}%
}
\end{figure}
\begin{figure}[ht!]
    \centering
    \includegraphics[width=0.6\textwidth]{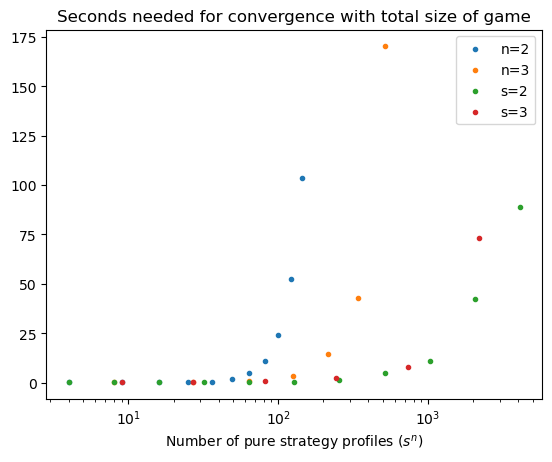}
    \caption{Convergence of our algorithm with total size of the game.}
    \label{fig:unif_conv_all}
\end{figure}

%% file: openquestions.tex
\section{Discussion and Open Questions}
We have proposed that a useful way of understanding a game in normal form is as a map between a prior distribution over mixed strategy profiles to a distribution over sink equilibria; namely, the distribution induced by the noisy replicator dynamics if started at the prior.  We showed that this distribution can be computed quite efficiently starting from any pure strategy profile, through a novel algorithm that handles the infinitesimal transitions associated with tie edges.  By implementing this algorithm and dynamical system we conducted experiments which we believe demonstrate the feasibility of this approach to understanding the meaning of a game. 
\\ \\
There are many problems left open by this work.
\begin{itemize}
\item In our simulations we approximated the meaning of the game for quite large games. We believe that more sophisticated statistical methods can yield more informative results for larger games. Another possible front of improvement in our simulations would be a better theoretical understanding of the trade-off between the parameters $\delta$ and  $\eta$ of the dynamics --- the length of the jump and the intensity of the noise.

\item Under which assumptions do the sink equilibria coincide with the chain recurrent components of the replicator dynamics (the solution concept suggested by the topological theory of dynamical systems)? Sharpening the result of Biggar and Shames in this way is an important open problem. On the other hand, a counterexample showing that it cannot be sharpened would also be an important advance; we note that experiments such as the ones in this paper are a fine way of generating examples of systems of sink equilibria which could eventually point the way to a counterexample.
Another question in the interface with the topological theory is, does the time average behavior within a sink SCC correspond to the behavior within a chain component of the replicator?

\item It would be very interesting to try --- defying PSPACE-completeness --- to compute sink equilibria and simulate the noisy replicator on succinct games such as extensive form, Bayesian, or graphical.
\end{itemize}